\documentclass[journal]{IEEEtran}

\usepackage{cite}
\usepackage{mathrsfs}
\usepackage{amsthm}
\usepackage[cmex10]{amsmath}
\usepackage{mathtools}
\usepackage{array}
\usepackage{eqparbox}
\usepackage{bm}
\usepackage{cases}
\usepackage{algorithm}
\usepackage{algpseudocode}

\usepackage{graphicx}
\usepackage{subfigure}
\usepackage{epstopdf}
\usepackage{epsfig}
\usepackage{amssymb}
\usepackage{array}
\usepackage{multirow}
\newtheorem{theorem}{Theorem}

\newtheorem{lemma}{Lemma}

\newtheorem{definition}{Definition}

\begin{document}
%
\title{Energy Efficient and Fair Resource Allocation for LTE-Unlicensed Uplink Networks: A Two-sided Matching Approach with Partial Information}

\author{Yuan Gao\textsuperscript{1}, Haonan Hu\textsuperscript{1},Yue Wu\textsuperscript{2*}, Xiaoli Chu\textsuperscript{1} and Jie Zhang\textsuperscript{1}
\thanks{This paper was supported in part by the Fundamental Research Funds for the Central Universities under Grant 222201814046; by the National Natural Science Foundation of China under Grant 61501187.}
\thanks{Yuan Gao, Haonan Hu, Xiaoli Chu and Jie Zhang are with the Department of Electronic and Electrical Engineering, The University of Sheffield, UK, e-mail: \{ygao41, haonan.hu, x.chu, jie.zhang\}@sheffield.ac.uk.}
\thanks{Yue Wu is with the School of Information Science and Engineering, East China University of Science and Technology, China, e-mail: yuewu@ecust.edu.cn.}
\thanks{\textsuperscript{*}The Email of the corresponding author is: yuewu@ecust.edu.cn}
}


%



\maketitle

\begin{abstract}
  LTE-Unlicensed (LTE-U) has recently attracted worldwide interest to meet the explosion in cellular traffic data. By using carrier aggregation (CA), licensed and unlicensed bands are integrated to enhance transmission capacity while maintaining reliable and predictable performance. As there may exist other conventional unlicensed band users, such as Wi-Fi users, LTE-U users have to share the same unlicensed bands with them. Thus, an optimized resource allocation scheme to ensure the fairness between LTE-U users and conventional unlicensed band users is critical for the deployment of LTE-U networks. In this paper, we investigate an energy efficient resource allocation problem in LTE-U coexisting with other wireless networks, which aims at guaranteeing fairness among the users of different radio access networks (RANs). We formulate the problem as a multi-objective optimization problem and propose a semi-distributed matching framework with a partial information-based algorithm to solve it. We demonstrate our contributions with simulations in which various network densities and traffic load levels are considered.
\end{abstract}


%
\IEEEpeerreviewmaketitle

\section{Introduction}
 1000x data requirement is a major challenge for cellular networks in 5G networks \cite{bleicher2013surge}.
 To overcome the challenge, exploiting more spectrums for reliable communication
 is regarded as a promising solution. Industrial scientific and medical (ISM) radio bands, in particular, 5.8 GHz have attracted wide interest \cite{chen2017coexistence}. The overall available
 spectrum bandwidth in the unlicensed bands in major markets (e.g. US, Europe, China, Japan)
 is several hundred megahertz (MHz)\cite{chen2017coexistence}.

LTE-unlicensed (LTE-U) is deployed to allow cellular user equipment (UE) to utilize ISM radio bands, in particular, 5.8 GHz. To enhance system capacity, unlicensed carriers are integrated
into a cellular network by using the carrier aggregation (CA). The CA enables the aggregation
of two or more component carriers into a combined bandwidth with one carrier serving as the Primary Component Carrier (PCC)
and others serving as Secondary Component Carriers (SCCs) \cite{wang2016evolution,lee2014recent,ku2015resource}. For LTE-U, licensed
carrier serves as the PCC, while the unlicensed bands work as
the SCCs in Time-Division-Duplexed (TDD) or Supplemental
DL (SDL) only \cite{chen2017coexistence}. Furthermore, in \cite{yued2dlteu}, the authors proposed a mechanism that allowed device-to-device (D2D) communications operating in unlicensed bands utilizing LTE-U technologies.

Wi-Fi networks, with low cost and high data rates, have been
the dominant players on all unlicensed bands in 2.4 and 5 GHz. However, spectrum
efficiency in Wi-Fi systems is low, especially given the
overloaded conditions. In contrast, LTE works more efficiently in terms of
resource management and error control. Therefore, the deployment of LTE-U
not only alleviates the spectrum scarcity of the cellular system, but also improves
the spectrum efficiency on the unlicensed bands.
\subsection{Challenges of Deploying LTE-U}
Despite the huge potential to meet cellular traffic surges,
LTE-U is still in its infancy; several deployment challenges remain to be overcome.
First, Wi-Fi systems would experience significant performance
degradation in the presence of LTE-U systems without a proper coexistence scheme \cite{babaei2015impact,rupasinghe2014licensed}.
Wi-Fi systems employ carrier sense multiple access with collision avoidance (CSMA/CA) to access the unlicensed bands, and a Wi-Fi user will back off if the
co-channel LTE-U signals is above the energy detection threshold (e.g., -62dBm over 20MHz) \cite{alcatel2016verizon}.
Therefore, a suitable coexistence mechanism is required in the LTE-U channel access scheme design.
Secondly, LTE-U users may fail to meet its quality of service (QoS) requirement due to Wi-Fi transmission. What's more,
the interference between LTE-U users of multiple operators would also lead to
performance degradation of LTE-U users.
Such unplanned and unmanaged deployment would result in severe
performance degradation for both Wi-Fi and LTE-U networks and poor spectrum efficiency. LTE-U
calls for coexistence schemes to enable harmonious resource sharing between Wi-Fi and LTE-U.

Thus, coexistence mechanisms have attracted substantial interest. Fair spectrum sharing between Wi-Fi and LTE-U can be ensured by using either non-coordinated or coordinated network managements. Non-coordinated schemes, such as LTE blank subframe allocation \cite{zhang2015coexistence}, listen-before talk (LBT) scheme \cite{chen2015downlink}, the carrier sense adaptive transmission (CSAT) by LTE-U forum \cite{alcatel2016verizon}, and 3 LBT schemes (Category (Cat) 2, 3, 4) by European Telecommunications Standards Institute (ETSI) \cite{ETSI1.7.1}, require modifications on the LTE-U side only, while coordinated schemes require information sharing about network operations and spectral resources using centralized network interconnections, including cooperative control for spectrum access and managing coexistence using an X2 interface \cite{al20155g}.

Research on the optimal resource allocation of the unlicensed spectrum has also been undertaken. Geometric programming \cite{chiang2005geometric} has been widely used in wireless communication to solve network resource allocation problems, which has been often used in LTE-U scenarios. In \cite{chen2016cellular}, the optimization performance of a hybrid method to perform both traffic offloading and resource sharing based on a duty cycle scheme is revealed. A fair-LBT (F-LBT) scheme is proposed by considering both the throughput and fairness of an LTE-U and a Wi-Fi system \cite{ko2016fair}. In \cite{gu2015exploiting}, a matching-based student-project model is developed to guarantee unlicensed users¡¯ QoS, together with the system-wide stability. Contention window size for both Wi-Fi and LTE-U users are jointly adapted to maximize LTE-U throughput
while guaranteeing the Wi-Fi throughput threshold \cite{gao}. In \cite{hamidouche2016multi}, power allocation problem of the small base stations is formulated as a non-cooperative game by using a multi-framework. Fair proportional allocation is developed to optimize both Wi-Fi and LTE-U throughput \cite{cano2015coexistence}. A centralized joint power optimization and joint time division channel access optimization scheme is proposed to achieve significant gains for both Wi-Fi and LTE-U throughput \cite{sagari2015coordinated}. A Nash bargaining game theoretic framework is also employed to solve the joint channel and power allocation problem in \cite{ni2012nash}. In \cite{chen2016optimizing}, the unlicensed spectrum is divided into a contention period, for Wi-Fi users only, and a contention-free period, for LTE-U users. The optimal contention period is obtained by using the Nash bargaining solution. In \cite{aliresource}, a joint user association and power allocation for licensed and unlicensed spectrum algorithm is proposed to maximize maximize sum rate of LTE-U/Wi-Fi heterogeneous networks.

Fair coexistence has not been defined clearly, and one of the definitions is that
the deployment of an LTE-U system should not affect one Wi-Fi system more than another
Wi-Fi system with respect to throughput and latency \cite{nielsen2014lte,chen2017coexistence}. Throughput fairness is explored by means of both $\alpha$-fairness and max-min approach and time division access and channel sharing between Wi-Fi and LTE-U are found to be effective coexistence schemes. Moreover, a criterion for switching between these two schemes is also established in \cite{garnaev2016fair}, subject to different network scenarios. We hold that a fair
coexistence should consider both Wi-Fi and LTE-U users' QoS, such as throughput
threshold and power consumption. Due to the limitations of power in end-user devices,
if a user's throughput requirement were fulfilled by consuming an excessive amount
of power, user's satisfaction would be affected. The ratio of the achievable
user throughput to the consumed user energy, i.e., energy efficiency (EE), is an important indicator for wireless communications especially
from a user's perspective, which has been widely explored in a 5G ultra-dense networks \cite{yang2017interference}, cognitive radio \cite{feng2013survey}, and OFDMA networks \cite{li2011energy}. Therefore, it is interesting and critical to study the EE minimization problem in Wi-Fi and LTE-U coexistence scenarios while meeting their QoS requirements.
\subsection{Matching Theory Framework}

Matching theory is a mathematical framework for forming mutually
beneficial relations, which was first applied in economics. It can be easily adapted to study resource allocation problems of a wireless communication system.
\begin{itemize}
\item Matching theory can model the interactions between two distinct sets of players
with different or even conflicting interests \cite{gu2015matching}. For example, in an LTE uplink network,
UE aims to achieve its QoS (mainly throughput) with minimal energy consumption while the objectives of small cell base stations (SCBSs)
are serving users with certain QoS requirements and maximizing its capacity.
\item Compared with game theory, a UE does not need other UEs' actions to make decisions. A preference list in terms of performance matrix, such as throughput and EE, is set up
based on the local information including channel conditions. UEs made proposals according to this list.
The only global information required from a centralized agent is the rejection/acceptance decision of each UE's proposal
and blocking pair.
\end{itemize}

Recently, matching theory has emerged as a promising tool to cope with future wireless
resource allocation problems. In a full duplex OFDMA network, UL and DL user pairing and
sub-channel allocations are modelled as a one-to-one three-sided matching to maximize the sum
system rate \cite{di2016joint}. In \cite{sekander2017decoupled}, the decoupled uplink-downlink user association problem in
multi-tier full-duplex cellular networks is formulated as two-sided many-to-one
matching. An algorithm, based on a stable marriage algorithm is
developed to find a near optimum with much lower complexity compared to a conventional
coupled and decoupled user association scheme. A resource allocation
problem for device-to-device (D2D) communications underlaying cellular networks is studied in \cite{zhao2017many};
a two-sided many-to-many matching scheme with an externality is proposed to find the sub-optimality.
A matching based algorithm to study the resource allocation problem in an LTE-U scenario is proposed
in \cite{gu2015exploiting}. The student-project model is used, in which students (cellular users)
propose projects (unlicensed bands), and the decisions are made by lectures (base stations) to
achieve maximal system (both LTE-U and Wi-Fi) throughput.
Based on this paper, the same goal is studied by considering user mobility in \cite{gu2017dynamic}.
However, all of the above work considers optimal system performance as a whole,
instead of QoS (such as throughput) for each user. In addition, another limitation of
the above works is that the matching is with complete preference lists. This is not
always the case in the real world, for example, some bands may fail to achieve a user's
QoS requirement, due to its availability and channel variation, which means that
some bands are not acceptable to certain users, making the preference list incomplete.

\subsection{Contributions}
\label{sec:contribution}

The major contributions of this paper are summarized as follows:

\begin{figure*}[htbp]
\centering\includegraphics[width=0.9\textwidth]{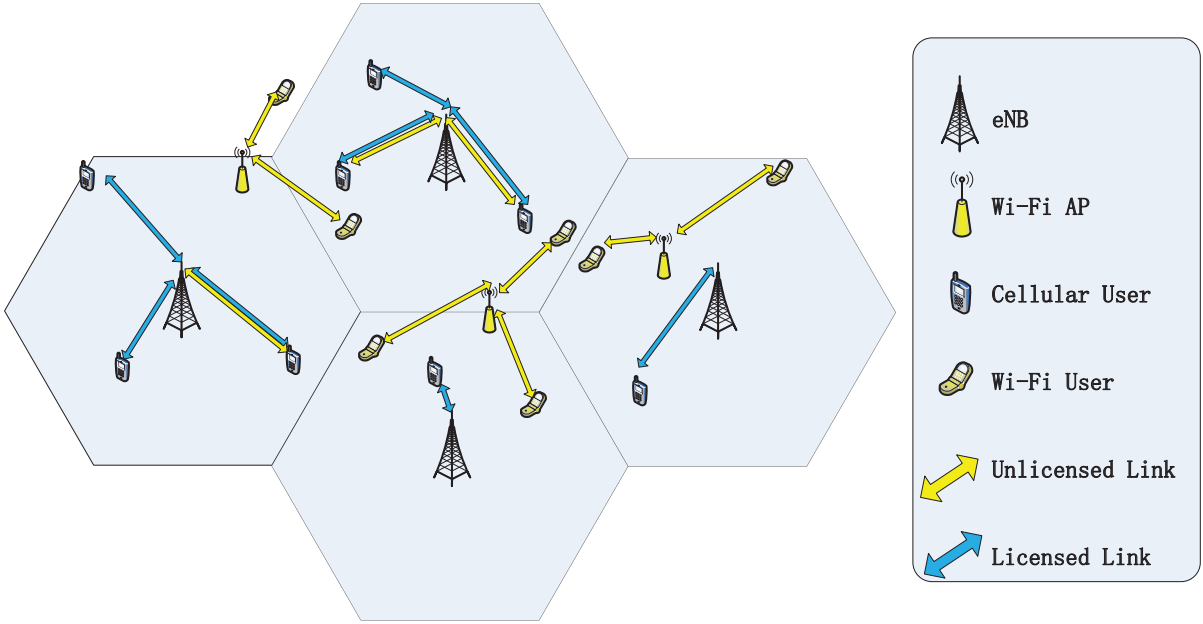}
\caption{System architecture of a LTE-U and Wi-Fi system}
\label{Scenario}
\end{figure*}
\begin{itemize}
\item Different from existing works, which typically consider only the fairness problem or overall EE (defined as the ratio of the overall data rate and the total energy consumption), we propose an optimized shared scheme for LTE-U networks coexisting with Wi-Fi in ISM bands, which aims at maximizing the EE of independent LTE-U users while guaranteeing fairness among different users. That is, the proposed algorithm would guarantee the QoS requirement for each user (including CUs and Wi-Fi users).
\item The optimization problem is formulated  as a \emph{multi-objective optimization problem}, in which typically a set of Pareto solutions can be achieved. We utilize the weighted sum method to transform the multi-objective optimization problem into a \emph{single-objective optimization problem} and find the Pareto optimal solution.
\item The single-objective optimization problem can be further modelled as a one-to-many matching game with partial information. Here \emph{partial information} means \emph{incomplete preference lists}, which is due to the fact that some UBs fail to fulfil a user's minimal throughput requirement and are not acceptable to that user. Such problem has not yet been solved. We propose a semi-distributed two-step stable algorithm to solve it. Numerical results demonstrate that the proposed algorithm can achieve good performance with fast convergence speed.
\end{itemize}

The rest of the paper is organized as follows. The system model is described in Section \uppercase\expandafter{\romannumeral3}. The problem formulation from a multi-objective optimization to a single-objective formulation is developed in \uppercase\expandafter{\romannumeral4}. To solve the optimization problem, a two-step matching-based resource allocation and user association algorithm are proposed in Section \uppercase\expandafter{\romannumeral5}. In Section \uppercase\expandafter{\romannumeral6}, numerical results are presented and analysed. Section \uppercase\expandafter{\romannumeral7} concludes the paper.
\section{System Model}
As shown in Fig. \ref{Scenario}, we consider a LTE-U network coexisting with a Wi-Fi network in ISM bands (e.g. 2.4 and 5.8 GHz), composed of $M$ independently uniformly distributed small-cell base stations (SCBSs), $SCBS=\{SCBS_1,...SCBS_m,...SCBS_M\}$, and $N$ independently uniformly distributed Wi-Fi access points (APs), $AP=\{AP_1,...AP_n,...AP_N\}$. All the SCBSs are deployed by the same cellular network operator. $K$ cellular users (CUs), $CU=\{CU_1,...CU_k,...CU_K\}$ and $N'$ Wi-Fi users (WU), $WU=\{WU_1,...WU_{n'},...WU_{N'}\}$  are uniformly distributed in the area of interest.

As shown in Fig. \ref{ABS}, the whole unlicensed spectrum is divided into $U$ orthogonal UBs. Then in the time domain, each UB is divided
into slots; the period of a slot is $T$. Each slot consists of several sub-frames, the duration of a subframe is $t$, which is
smaller than the coherence time of the signal channel. Thus, during the transmission period of a sub-frame, the power attenuation caused by Rayleigh fading in each link can be regarded as a fixed parameter. Moreover, each sub-frame is considered strictly independent.
\begin{figure}[htbp]
\centering\includegraphics[width=0.52\textwidth]{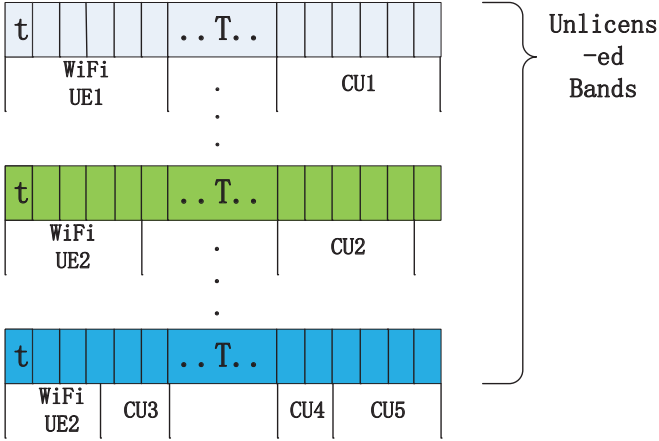}
\caption{TDD sharing of unlicensed bands between Wi-Fi and LTE-U users}
\label{ABS}
\end{figure}

WUs communicate with Wi-Fi APs under a standard carrier sense multiple access protocol with collision avoidance (CSMA/CA). CUs are served by SCBSs by using a licensed band for both uplink and downlink transmission, while they seek to aggregate unlicensed bands for a supplementary uplink transmission.

\renewcommand\arraystretch{1.4}
\begin{table}
  \begin{center}
  \caption{General Notation}
    \label{General_Notation}
    \begin{tabular}{ | l| l |}
   \hline
$SCBS_m$         & the $m$th small cell base station \\   \hline
$AP_n$           & the $n$th access point\\   \hline
$CU_k$           & the $k$th cellular user\\   \hline
$UB_u$           & the $u$th unlicensed band  \\   \hline
$T$    & slot time\\   \hline
$t$              & sub-frame time\\   \hline
$l_u$            & the fraction of time LTE-U is muting on $UB_u$ \\   \hline
\multirow{2}{*}{$C^C_{k,m,u}$}    &the uplink capacity $CU_k$ associating with $SCBS_m$ \\&on unlicensed band $UB_u$\\   \hline
\multirow{2}{*}{$I_{k,m,u}$}     & the number of sub-frames in $UB_U$ allocated to $CU_k$\\& served by $SCBS_m$\\   \hline
$C_{k,m,u,i}$    &the achievable data rate of $CU_k$ served by$SCBS_m$\\   \hline
$\chi_{k,m,u}$   & equals 1 if $CU_k$ is served by $SCBS_m$ using $ UB_u$ \\   \hline
$P_{k,m}^{CU}$   & transmission power from $CU_k$ to $SCBS_m$\\   \hline
\multirow{2}{*}{$g_{k,m,u}$}      & channel power gain between $CU_k$ and $SCBS_m$ \\&on $UB_u$\\   \hline
\multirow{2}{*}{$R_{k,m,u}$ }     & the uplink throughput of $CU_k$ served \\&by $SCBS_m$ on $UB_u$\\   \hline
$\sigma_{N}^{2}$ &the thermal noise\\   \hline
$WU_u$           & Wi-Fi users on $UB_U$\\   \hline
$R^W_u$          &throughput requirement of $WU_u$ \\   \hline
${PE}_k^{CU}$    &energy efficiency of $CU_k$\\   \hline
$R_k^{L}$        & Throughput requirement of $CU_k$\\   \hline

\hline
\end{tabular}
\end{center}
\end{table}

A CU can access its local SCBS for uplink transmission with one of $U$ UBs. We consider LTE-U using a duty cycle scheme to manage the coexistence in the unlicensed spectrum in the time domain. By using this duty cycle method, CUs will use a almost blank subframe (ABS) pattern
\cite{zhang2015coexistence} to guarantee Wi-Fi QoS by muting a fraction of time for $UB_u$. The fraction $l_u$ will be adaptively adjusted based on the Wi-Fi data requirement. Here, we consider the static synchronous muting pattern.

The notations in this paper can be found in Table \ref{General_Notation}.

%
\subsection{LTE-U Throughput}
During the transmission slot of LTE-U, we denote the uplink capacity  $C_{k,m,u}^{C}$ of $k$-th CU $CU_k$ associating with $SCBS_m$ on unlicensed band $UB_u$. Thus, the uplink throughput on $UB_u$ is given by:
\begin{small}
\begin{equation}
R_{k,m,u}^{CU}=\sum_{i=1}^{I_{k,m,u}}C_{k,m,u,i}^{CU},
\end{equation}
\end{small}

where $I_{k,m,u}$ is the number of sub-frames in $UB_U$ allocated to $CU_k$ served by $SCBS_m$. $C_{k,m,u,i}$ is the achievable data rate of $CU_k$ served by $SCBS_m$
the $u$-th sub-frame of $UB_u$, given as:
\begin{small}
\begin{equation}
C_{k,m,u,i}^{CU}= {t_i}B_{u}log_{2}(1+\frac{\chi_{k,m,u}P_{k,m}^{CU}g_{k,m,u}}{\sigma_{N}^{2}+\sum_{j\neq k}^{K}\sum_{m}^{M}\rho_{j,m,u}P_{j,m}^{CU}g_{j,m,u}})
\end{equation}
\end{small}where, $\chi_{k,m,u}$ is an indicator function, defined as:
\begin{small}
\begin{equation}
\chi_{k,m,u}=\begin{cases}
 & \text{1, if } CU_k \text{ is served by } SCBS_m \text{ using } UB_u,\\
 & \text{0, otherwise.}
\end{cases}
\end{equation}
\end{small}$P_{k,m}^{CU}$ represents the transmission power from $CU_k$ to $SCBS_m$. $g_{k,m,u}$ is the channel power gain between $CU_k$ and $SCBS_m$ on $UB_u$, and $g_{j,m,u}$ is the channel gain between $CU_j$ and $SCBS_m$ on $UB_u$. $\sigma_{N}^{2}$ is the thermal noise.
\subsection{Wi-Fi Throughput}
For each WU $WU_{n'}$, there is equal probability of accessing one of the unlicensed bands. We regard the WUs sharing the same UB as one WU, the interactions between co-channel CUs and WUs can be simplified to the interactions between co-channel CUs and a WU\cite{gu2015exploiting,gu2017dynamic}. The WU that occupies $UB_u$ is denoted as $WU_u$. Thus, the throughput of $Th_u$ can be expressed by \cite{bianchi2000performance}:
\begin{small}
\begin{equation}
Th_u=\frac{\overline{E(p)}P_{tr}^{u}P_{s}^{u}}{(1-P_{tr}^{u})\delta+P_{tr}^{u}P_{s}^{u}{T_{s}}+P_{tr}^{u}(1-P_{s}^{u}){T_{c}}},
\end{equation}
\end{small}where $\overline{E(p)}$ is the average packet size of Wi-Fi transmission, $P_{tr}^{u}$ is the probability that $UB_u$ is occupied, and $P_{s}^{u}$ is the successful transmission probability in $UB_u$. $\delta$ is the slot time defined in 802.11. ${T_{s}}$ and ${T_{c}}$ are the average time consumed by a successful transmission and a collision in $UB_u$, respectively.

Based on the ABS scheme, the fraction of time slots $l_u$ of $UB_u$ will be allocated to the $WU_u$ using $UB_u$. To guarantee throughput requirement $R^W_u$ of $WU_u$, $l_u$ is given as£º

\begin{small}
\begin{equation}
Th_ul_{u}T\geq R_u^W.
\end{equation}
\end{small}
\section{Problem Formulation}
We define the EE of $CU_k$, i.e., the throughput of $CU_k$ obtained per unit power consumption with the unit of '$bits-per_Joule$' \cite{feng2013survey} as follows:
\begin{equation}
  {PE}_k^{CU}=\frac{\sum_m^M\sum_u^U\chi_{k,m,u}R_{k,m,u}}{{\sum_m^M\sum_u^U\chi_{k,m,u}I_{k,m,u}P_{k,m}^{CU}}}
\end{equation}
We formulate the following EE maximization problem for each CU as a multi-objective optimization problem:
\begin{subequations}
\begin{align}
&{min}(-{PE}_{1}^{CU},...,-{PE}_{K}^{CU})\tag{7}\label{eq:2},\\
s.t \nonumber\\
&\sum_{k}^{K}\sum_{u}^{U}\chi_{k,m,u}\leq 1,\, m \in \{1, ..., M\} \label{eq:3a},\\
&\sum_{m}^{M}\sum_{u}^{U}\chi_{k,m,u} I_{k,m,u} t\leq T{l_u},\,\,\,\, k \in \{1, ..., K\} \label{eq:3b},\\
&\chi_{k,m,u}\in \left \{0,1 \right \} ,\,k \in \{1, ..., K\}, m \in \{1, ..., M\}, \nonumber\\&\,\,\,\,\,\,\,\,\,\,\,\,\,\,\,\,\,\,\,\,\,\,\,\,\,\,\,\,\,\,\,\,\,\,\,\,\,\,\,\,\,\,\,\,\,u \in \{1, ..., U\}\label{eq:3c},\\
&P_{k,m}^{CU}\leq P_{max},\,k \in \{1, ..., K\}, m \in \{1, ..., M\}\label{eq:3d},\\
&Th_u(l_{u})T\geq R_u^W,\, u \in \{1, ..., U\}\label{eq:3e},\\
&\sum_{m}^{M}\sum_{u}^{U}\chi_{k,m,u}R_{k,m,u}\geq R_k^{L},\, k \in \{1, ..., K\} \label{eq:3f}.
\end{align}
\end{subequations}where, constraint (\ref{eq:3a}) indicates that a CU can be allocated up
 to 1 UB at a time. (\ref{eq:3b}) is the limitation of the available
 resource of each UB for LTE-U transmission. In (\ref{eq:3c}),
 $\chi_{k,m,u}$ is a binary number, equal to 1 if $CU_k$ served
 by $SCBS_m$ on $UB_u$, or 0 otherwise. The transmission power limit
 of each CU is set in (\ref{eq:3d}). The throughput minimum requirement
 of each Wi-Fi user and CU is shown in (\ref{eq:3e}) and (\ref{eq:3f}),
 respectively.

The general technique used to solve the multi-objective optimization is a weighted-sum or scalarization method by transforming a multi-objective function into a single-objective function \cite{ruzika2005approximation} as:
\begin{subequations}
\begin{align}
&{min}({-\sum_{k=1}^K}\gamma_k {PE}_k^{CU})\tag{8}\label{eq:3},\\
s.t \nonumber\\
&\sum_{k=1}^K\gamma_k=K,\\
&\sum_{k}^{K}\sum_{u}^{U}\chi_{k,m,u}\leq 1,\, m \in \{1, ..., M\} \label{eq:4a},\\
&\sum_{m}^{M}\sum_{u}^{U}\chi_{k,m,u} I_{k,m,u} t\leq T{l_u},\,\,\,\, k \in \{1, ..., K\} \label{eq:4b},\\
&\chi_{k,m,u}\in \left \{0,1 \right \} ,\,k \in \{1, ..., K\}, m \in \{1, ..., M\}, \nonumber\\&\,\,\,\,\,\,\,\,\,\,\,\,\,\,\,\,\,\,\,\,\,\,\,\,\,\,\,\,\,\,\,\,\,\,\,\,\,\,\,\,\,\,\,\,\,u \in \{1, ..., U\}\label{eq:4c},\\
&P_{k,m}^{CU}\leq P_{max},\,k \in \{1, ..., K\}, m \in \{1, ..., M\}\label{eq:4d},\\
&Th_u(l_{u})T\geq R_u^W,\, u \in \{1, ..., U\}\label{eq:4e},\\
&\sum_{m}^{M}\sum_{u}^{U}\chi_{k,m,u}R_{k,m,u}\geq R_k^{L},\, k \in \{1, ..., K\} \label{eq:4f}.
\end{align}
\end{subequations}
The effectiveness of the transformations is given in $Lemma$ \ref{multi-single} \cite{ruzika2005approximation} as:
\begin{lemma}\label{multi-single}
The single-objective minimizer is an effective solution for the original multi-objective problem.
If the $\gamma_k$ weight vector is strictly greater than zero, then the single-objective
minimizer is a strict Pareto optimum.
\end{lemma}

where strict Pareto optimum is defined as follows:
\begin{definition}\label{optimum}
Strict Pareto Optimum: A solution Matrix \textbf{\textit{M}} is said to be a strict Pareto optimum or a strict efficient solution for the
multi-objective problem (\ref{eq:2}) if and only if there is no $m\subseteq S$ such that ${PE}_k^{CU}(m)\leq  {PE}_k^{CU}(m')$
for all $k \in  {1, . . . ,K}$, with at least one strict inequality. $S$ is the constraints (\ref{eq:3a}-\ref{eq:3f}).
\end{definition}

If all the CUs are of the same priority, i.e.,
\begin{equation}\label{123}
  \gamma_k=1, k \in \{1, ..., K\}.
\end{equation}
The EE optimization is finally transformed into:
\begin{subequations}
\begin{align}
&{min}(-\sum_{k=1}^K {PE}_k^{CU})\tag{10}\label{eq:5},\\
s.t \nonumber\\
&\sum_{k}^{K}\sum_{u}^{U}\chi_{k,m,u}\leq 1,\, m \in \{1, ..., M\} \label{eq:5a},\\
&\sum_{m}^{M}\sum_{u}^{U}\chi_{k,m,u} I_{k,m,u} t\leq T{l_u},\,\,\,\, k \in \{1, ..., K\} \label{eq:5b},\\
&\chi_{k,m,u}\in \left \{0,1 \right \} ,\,k \in \{1, ..., K\}, m \in \{1, ..., M\}, \nonumber\\&\,\,\,\,\,\,\,\,\,\,\,\,\,\,\,\,\,\,\,\,\,\,\,\,\,\,\,\,\,\,\,\,\,\,\,\,\,\,\,\,\,\,\,\,\,u \in \{1, ..., U\}\label{eq:5c},\\
&P_{k,m}^{CU}\leq P_{max},\,k \in \{1, ..., K\}, m \in \{1, ..., M\}\label{eq:5d},\\
&Th_u(l_{u})T\geq R_u^W,\, u \in \{1, ..., U\}\label{eq:5e},\\
&\sum_{m}^{M}\sum_{u}^{U}\chi_{k,m,u}R_{k,m,u}\geq R_k^{L},\, k \in \{1, ..., K\} \label{eq:5f}.
\end{align}
\end{subequations}

We denote the solution for optimization problem (\ref{eq:5}) as Matrix \textbf{\textit{M}}, which, according to Lemma. \ref{multi-single},
is an strict Pareto optimum for the multi-objective optimization problem (\ref{eq:2}).

In the expression of $PE_k^{CU}$, which is nonlinear, $I_{k,m,u}$ and $\chi_{k,m,u}$ are integers, while $R_{k,m,u}$ and $P_{k,m}^{CU}$ are continuous variables. The objective function (10) is a summation of $PE_k^{CU}, k \in \{1, ..., K\}$, thus, it is a mixed integer nonlinear programming (MINLP) problem, which is typically NP-hard. Thus, to reduce the computation complexities, we developed a matching-based solution, which will be discussed in the following section.
\section{Matching with Incomplete Preference Lists}
\subsection{Introduction to Matching Theory and Student-Project-Allocation Problem}
Student project allocation (SPA) is a one-to-many matching game, where each student has a preference list of the projects that they can choose from, while the lecturers have a preference list of students for each project or a preference list of student-project pairs. There is an upper bound, also known as the quota, on the number of students  that can be assigned to each particular project \cite{manlove2013algorithmics}.

Inspired by the SPA problem, we model the resource allocation problem in (\ref{eq:5}) as an SPA game, where the CUs, UBs and SCBSs are considered equivalent to students, projects and lecturers, respectively. Similarly, SCBSs offer the set of available UBs and maintain a preference list for each UB, and each CU has a preference list of UBs that they can use for uplink transmission. SCBSs allocate UBs to CUs based on the achievable EE on UBs. Meanwhile, our resource allocation problem differs from the SPA game in the following aspects:
\begin{itemize}
\item \textbf{Maximum throughput}:
The quota in the SPA problem is replaced by the maximum achievable throughput of a UB. The maximum achievable throughput of a UB determines the maximum number of CUs that it can be allocated to while meeting the minimum required Wi-Fi throughput in the TDD mode.
\item \textbf{Incompleteness of preference lists}:
The SCBSs sense the availabilities of and keep the CUs updated. Any UB that is not able to fulfil a CU's minimal throughput requirement will be deleted from the preference list of the CU and the CU will be removed from the preference list of that UB. Only a subset of UBs (CUs) are in the preference list of a CU (UB), i.e., the preference lists are incomplete.
\end{itemize}

The \textit{k}th CU preferring the \textit{u}th UB over the \textit{u'}th UB is denoted by $pri(CU_k, UB_u)> pri(CU_k, UB_{u'})$. Similarly,
$pri(UB_u, CU_k)> pri(UB_u, CU_{k'})$ indicates that the \textit{u}th UB prefers the
\textit{k}th CU over \textit{k'}th CU. The one-to-many matching is defined as follows:
\begin{definition}
Let \textbf{$\mu$} denote the one-to-many matching between two disjoint sets \textbf{\textit{CU}} and \textbf{\textit{UB}}.

$\mu(CU_k)=UB_u$ indicates that the \textit{kth} CU is matched to the \textit{uth} UB,

$\mu(UB_u)=\{CU_k,..., CU_{k'}\}$ indicates that the \textit{uth} UB is matched to $\{CU_k,..., CU_{k'}\}$,

$\mu(CU_k)=CU_k$ indicates that the \textit{kth} CU is not really matched to any UB.
\end{definition}
The stability implies the robustness of the matching against deviations caused by the individual rationality of players, i.e., the CUs in our resource allocation problem. In an unstable matching, two CUs may swap their matched UBs to maximize their own EE, leading to an undesirable and unstable resource allocation. The definition of stability of the one-to-many matching is given as follows:

\begin{definition}{Stability of One-to-Many Matching}.
The one-to-many matching \textbf{$\mu$} between two disjoint sets \textbf{\textit{CU}} and \textbf{\textit{UB}} is stable, only if it is not blocked by any blocking individual or blocking pair, where the blocking individual and the blocking pair are defined in the following.
\end{definition}

Blocking individual in the EE optimization problem is defined as:
 \begin{definition}{Blocking Individual}.
 A CU is a blocking individual if it prefers to stay unmatched rather than being matched to any available UB.
\end{definition}

The blocking pair in the EE optimization problem is defined as:
\begin{definition}{Blocking Pair}.
A pair $(CU_k, UB_u)$ is a blocking pair if all the following 3 conditions are satisfied:

(1) $\mu(CU_k)$$\neq $$UB_u$ and $pri(CU_k,UB_u)$$>$$pri(CU_k,\mu(CU_k))$;

(2) $\mu(UB_u)$$\neq $$CU_k$ and $pri(UB_u,CU_k)$$>$$pri(UB_u,\mu(UB_u))$;

(3) There is enough spectrum in $UB_u$ to meet the minimum throughput requirement of $CU_k$.
\end{definition}
\subsection{Preference Lists of CUs Over UBs}
We assume that the preference of $CU_k$ over $UB_u$ is based on EE $PE_{k,m,u}^{CU}$ achieved by $CU_k$ served by $SCBS_m$ using $UB_u$ to guarantee its QoS threshold, which is written as follows:
\begin{equation} PE_{k,m,u}^{CU}=\frac{\sum_m^M\sum_u^U\chi_{k,m,u}R_{k,m,u}}{{\sum_m^M\sum_u^U\chi_{k,m,u}P_{k,m}^{CU}I_{k,m,u}t}}
\end{equation}
$CU_k$ prefers $UB_{u}$ over $UB_{u'}$ if $CU_k$ can achieve higher EE
using $UB_{u}$ than $UB_{u'}$, which is stated as follows:
\begin{equation}
pri(CU_k,UB_{u})>pri(CU_k,UB_{u'}) \Leftrightarrow PE_{k,m,u}^{CU}>PE_{k,m,u'}^{CU}
\end{equation}
None of the CUs have any knowledge about other co-channel coexisting CUs, before the final band allocation is
performed at SCBSs. Thus, the preference lists are set up based on local channel sensing
information and unlicensed band availability alone.
\subsection{Preference Lists of SCBS Over ($CU_k$, $UB_u$) Pair}
However the preference of $SCBS_m$ over the user-band pair ($CU_k$, $UB_u$) is based on the EE achieved by allocating $UB_u$ to $CU_k$ to fulfil the QoS threshold of $CU_k$. It is written as $SCBS_m$ prefers $CU_k$ over $CU_{k'}$ to occupy $UB_{u}$ if $CU_k$ can achieve higher EE
than $CU_{k'}$ by using $UB_{u}$, which is stated as follows:
\begin{equation}
pri(UB_{u},CU_k)>pri(UB_{u},CU_{k'}) \Leftrightarrow PE_{k,m,u}^{CU}>PE_{k',m,u}^{CU}
\end{equation}

\subsection{Two-Step Algorithm}
\subsubsection{Step 1: Modified GS Algorithm for One-to-Many Game}
\label{sec:step1}
To solve the above matching game, a 2-step algorithm is proposed.
The first step is an extension of the GS algorithm applied for a one-to-many matching with
incomplete preference lists. Each iteration begins with the unmatched
CUs proposing their favourite (i.e., the first UB) UB on their current preference lists.
The UBs which have been proposed to will be removed from the CUs' preference lists. For each $UB_u$, SCBSs
decide whether to accept or reject the CU's proposal $UB_u$ based on SCBSs' preference lists over ($CU_k$, $UB_u$) pairs.
SCBSs choose to keep the most preferred CUs as long as these CUs do not occupy more resources
than the UB could offer; the remaining CUs are rejected. Such a procedure runs until every CU is either matched or its preference list is empty. The
implementation detail of Step 1 of the algorithm is stated in \ref{alg:step1}
as follows:

\begin{algorithm}
    \caption{One-to-Many Matching}
    \label{alg:step1}
    \begin{algorithmic}[1]
        \State \textbf{Input:} $CU$, $UB$, $PL^{CU}$, $PL^{UB}$
        \State \textbf{Output:} Matching \textbf{$\mu_1$}
        \State \textbf{Step 1£º} Proposing£º
        \State \quad All free $CU_k$ propose their favourite $UB_u$ in their preference lists, and remove $UB_u$ from the list.
        \State \textbf{Step 2£º} Accepting/rejecting£º
       \State \quad $UB_u$ accept the most preferred $n$ proposers based on its preference list, the rest are rejected. The sum of the slot time of the accepted proposers does not exceed its available resource time.
       \State \quad None of the accepted proposers are free.
       \State \quad All the rejected proposers are free.
\State \textbf{Criterion£º}
\State \quad If every CUs is either allocated with a UB or its preference list is empty, this algorithm is terminated with an output $M_1$.
\State \quad Otherwise, \textbf{Step 1} and \textbf{Step 2} are performed again.
   \end{algorithmic}
\end{algorithm}
\begin{theorem}{Stability of \textbf{$\mu_1$}}.
\label{sta1}
In any instance of one-to-many matching, stable matching is achieved by using \ref{alg:step1}.
\end{theorem}
\begin{proof}
We prove this theorem by contradiction and assume that for an instance of one-to-many matching, \ref{alg:step1} terminates with an instable matching \textbf{$\mu_1$}, i.e., there exists at least one blocking pair ($CU_k$, $UB_u$) or one blocking individual $CU_k$.

If there exists one blocking pair ($CU_k$, $UB_u$) in \textbf{$\mu_1$}:
\begin{itemize}
\item Case 1: In \textbf{$\mu_1$}, ${UB}_u$ is unmatched and ${CU}_k$ is matched with ${UB}_u'$.
If ${UB}_u$ is not on the preference list of ${CU}_k$, then, ${CU}_k$ does not have an incentive to match with ${UB}_u$; If $pri(CU_{k},UB_{u'})>pri(CU_k,UB_{u})$, and ${CU}_k$ is matched with ${UB}_u'$ in \textbf{$\mu$}, then ${CU}_k$ does not have an incentive to match with ${UB}_u$; If $pri(CU_{k},UB_{u})>pri(CU_k,UB_{u'})$, then ${CU}_k$ proposes to ${UB}_u$ before ${UB}_{u'}$. ${CU}_k$ is rejected during the proposal stage or is accepted by ${UB}_u$ first, then is rejected. In conclusion, in any situation in which ${CU}_k$ is matched and ${UB}_u$ is unmatched, a blocking pair does not exist.
\item Case 2: In \textbf{$\mu_1$}, ${UB}_u$ being unmatched and ${CU}_k$ unmatched. ${UB}_u$ is unmatched means that it receives no proposal from CU, including ${CU}_k$. This means that ${UB}_u$ is not on ${CU}_k's$ preference list, then ${CU}_k$ does not have incentive to match with ${UB}_u$. In conclusion, in any situation in which both ${CU}_k$ and ${UB}_u$ are unmatched, blocking pair does not exist.

\item Case 3: In \textbf{$\mu_1$}, ${UB}_u$ being matched with ${CU}_k'$ and ${CU}_k$ unmatched. ${CU}_k$ is unmatched means that either it has no ${UB}_u$ in its preference list, or all its proposals have been rejected. For the former, ${CU}_k$ does not have an incentive to match with ${UB}_u$. For the latter, ${UB}_u$ rejects ${CU}_k$ because it prefers other proposer(s). Thus, ${UB}_u$ does not have an incentive to match with ${CU}_k$. In conclusion, in any situation in which both ${CU}_k$ is unmatched and ${UB}_u$ is matched, blocking pair does not exist.
\item Case 4: In \textbf{$\mu_1$}, ${UB}_u$ is matched with ${CU}_k'$ and ${CU}_k$ with ${UB}_u'$. ${UB}_u$ must be on ${CU}_k'$s preference list, and vice versa, otherwise, there is no incentive to form the (${CU}_k$, ${UB}_u$) pair. If $pri(CU_{k},UB_{u'})>pri(CU_k,UB_{u})$, then, ${CU}_k$ does not have an incentive to match with ${UB}_u$ if it is matching with ${UB}_{u'}$. If $pri(CU_{k},UB_{u})>pri(CU_k,UB_{u'})$, then, ${CU}_k$ proposes to ${UB}_u$ first and is rejected, because ${UB}_u$ prefers ${CU}_k'$ to ${CU}_k$, then ${UB}_u$ does not have an incentive to match with ${CU}_k'$. In conclusion, in any situation in which both ${CU}_k$ and ${UB}_u$ are matched, a blocking pair does not exist.
    \end{itemize}
Contradictions, as (${CU}_k$, ${UB}_u$) is any pair, thus, it could be said that there is no blocking pair in matching \textbf{$\mu_1$}.

If one blocking individual $CU_k$ or $UB_u$ exists in \textbf{$\mu_1$}:

for blocking individual $CU_k$:
\begin{itemize}
\item Case 1: In \textbf{$\mu_1$}, ${CU}_k$ is matched with ${UB}_u$, i.e., ${UB}_u$ is on ${CU}_k$'s preference list, as such ${CU}_k$ does not have incentive be unmatched. In conclusion, in any situation in which both ${CU}_k$ and ${UB}_u$ are unmatched, blocking individual $CU_k$ does not exist.
    \end{itemize}

The proof that blocking individual $UB_u$ does not exist is similar to that blocking individual $CU_k$ does not exist.

As the above blocking pair ($CU_k$, $UB_u$), blocking individuals $CU_k$ or $UB_u$ can be any pair or individual, thus, we could prove that there is no blocking pair or blocking individual in matching \textbf{$\mu_1$}.
\end{proof}

\begin{theorem}{Praeto optimality of \textbf{$\mu_1$}}.

In any instance of one-to-many matching, stable matching \textbf{$\mu_1$} achieved by \ref{alg:step1} is Praeto optimal, i.e., no player(s) can better off, whilst no players are worse off.
\end{theorem}
\begin{proof}
In stable matching \textbf{$\mu_1$}:
\begin{itemize}
\item Case 1: There exists an unmatched ${CU}_k$, which can be matched to ${UB}_u$ to increase the achievable EE of both ${CU}_k$ and ${UB}_u$, meaning that (${CU}_k$, ${UB}_u$) is the blocking pair of matching \textbf{$\mu_1$}, contracting \textbf{Theorem \ref{sta1}}.

\item Case 2: There exists a (${CU}_k$, ${UB}_u$) pair. Obviously, ${CU}_k$ does not have an incentive to be unmatched; ${CU}_k$ has the incentive to change partner from ${UB}_u$ to ${UB}_{u'}$ to increase its achievable EE, meaning that (${CU}_k$, ${UB}_{u'}$) is a blocking pair of matching \textbf{$\mu_1$}, contracting \textbf{Theorem \ref{sta1}}.
    \end{itemize}
    It is impossible to increase the EE of some CUs' without decreasing that of the remaining of the CUs. The state stands for UB, which can be proven similarly as above.
\end{proof}

We define the computational complexity of \ref{alg:step1} as the number of accepting/rejecting decisions required to output a stable matching \textbf{$\mu_1$}. The complexity of \ref{alg:step1}, i.e., the convergence of \ref{alg:step1} is given in \textbf{Theorem \ref{comSTEP1}}.

\begin{theorem}\label{comSTEP1}{Complexity of \ref{alg:step1} (Convergence of \ref{alg:step1})}.
\label{complexity}
In any instance of many-to-one matching, a matching \textbf{$\mu_1$} can be obtained by using \ref{alg:step1} within $\mathcal{O}(KU)$ iterations.
\end{theorem}
\begin{proof}
In each iteration, a CU proposes to its most favourite UB in its current preference list, and SCBS accepts/rejects the proposal. The maximum number of elements in the preference list of $CU_k$ equals the number of UBs, i.e., $U$. Thus, stable matching \textbf{$\mu_1$} can be obtained in $\mathcal{O}(KU)$ overall time, where $K$ is the number of CUs and $U$ is the number of UBs.
\end{proof}
\subsubsection{Step 2: EE Optimization}
As proven above, stability and Pareto optimality have been guaranteed by using algorithm \ref{alg:step1},
meaning that there are no incentives for any CUs and UBs to form new matching.
However, the preference lists of CUs could to be incomplete, some CUs may be unmatched \cite{iwama2008survey}, \cite{shrivastava2016stable}.

To further maximize system's EE by increasing the number
of CUs matched by algorithm \ref{alg:step2}, an iteration of algorithm \ref{alg:step2}
begins with an unmatched $CU_k$ proposing to its most favourite $UB_u$, and $UB_u$
would be deleted from the preference list of $CU_k$. An SCBS would
consider this proposal acceptable if the following criteria are fulfilled:
\begin{itemize}
\item After deleting several non-favourites or all CUs matched with $UB_u$ in \textbf{$\mu_1$} obtained via algorithm \ref{alg:step1},
the minimal throughput of $CU_k$ can be achieved by using $UB_u$
\item All the deleted CUs could be served by other UBs to fulfil
their minimal throughput requirement.
\item The EE of the new matching \textbf{$\mu_k$} is greater than
that of the previous matching \textbf{$\mu_1$}.
\end{itemize}
Such matching \textbf{$\mu_k$} would be considered as a profitable reallocation, and would be
updated as the new matching, if only one profitable reallocation exists.
Should there be multiple profitable reallocations, the one that enhances
the overall EE the most would be the new matching. The iterations
would run several times, until every CU is either allocated with a UB or its
preference list is empty. The detail of algorithm \ref{alg:step2} is described as follows:

\begin{algorithm}
    \caption{System EE Maximization}
    \label{alg:step2}
    \begin{algorithmic}[1]
        \State \textbf{Input:} $CU$, $UB$, $PL^{CU}$, $PL^{UB}$, \textbf{$\mu_1$}
        \State \textbf{Output:} Matching \textbf{$\mu_2$}
        \State \textbf{Step 1£º} Proposing£º
        \State \quad Every free $CU_k$ proposes to their favourite $UB_u$ in their preference lists, and removes $UB_u$ from the list.
        \State \textbf{Step 2£º} Reallocation£º
       \State \quad Each $CU_k$ is accommodated in $UB_u$ by deleting its non-favorite partners in \textbf{$\mu_2$},
       to ensure that the occupying slot time does not exceed the available slot time
       \State \quad All the deleted CUs can be accommodated by other UBs. A matching \textbf{$\mu_k$} is formed.
       \State \quad EE increases from matching \textbf{$\mu_1$} to \textbf{$\mu_k$}.
       \State \quad \textbf{$\mu_k$} is stored if all the above three criteria are fulfilled. \textbf{Step 2} is performed
       until all free CUs have gone through \textbf{Step 2}.
       \State \textbf{Step 3£º} Accepting/rejecting£º
       \State \quad The \textbf{$\mu_k$} that increases the system's EE most
       is updated; $CU_k$ is set to be served. The rest \textbf{$\mu_{k'}$} are rejected,
       and $CU_{k'}$ are rejected and set to be free.
\State \textbf{Criterion£º}
\State \quad Each CUs is either allocated with a UB or its preference list is empty, this algorithm is terminated with an output \textbf{$\mu_2$}.
\State \quad Otherwise, \textbf{step 1}, \textbf{step 2} and \textbf{step 3} are performed again.
   \end{algorithmic}
\end{algorithm}

\begin{theorem}{Stability of \textbf{$\mu_2$}}.
\label{sta2}
In any instance of one-to-many matching, stability is achieved by using \ref{alg:step2} in \textbf{$\mu_2$}.
\end{theorem}
\begin{proof}
We prove this theorem by contradiction and assume that for an instance of one-to-many matching, \ref{alg:step2} terminates with an instable matching \textbf{$\mu_2$}, i.e., there exists at least one blocking pair ($CU_k$, $UB_u$) or one blocking individual $CU_k$ or $UB_u$.

If there exists one blocking pair ($CU_k$, $UB_u$) in \textbf{$\mu_2$}:
\begin{itemize}
\item Case 1: In \textbf{$\mu_2$}, ${UB}_u$ is unmatched and ${CU}_k$ is matched with ${UB}_u'$.
If ${UB}_u$ is not on the preference list of ${CU}_k$, then, ${CU}_k$ does not have an incentive to match with ${UB}_u$; If $pri(CU_{k},UB_{u'})>pri(CU_k,UB_{u})$, and ${CU}_k$ is matched with ${UB}_u'$ in \textbf{$\mu_2$}, then ${CU}_k$ does not have an incentive to match with ${UB}_u$; If $pri(CU_{k},UB_{u})>pri(CU_k,UB_{u'})$, then ${CU}_k$ proposes ${UB}_u$ before ${UB}_{u'}$ in \ref{alg:step1}, or re-matches to ${UB}_u$ before ${UB}_{u'}$ in \ref{alg:step2}. The result is that ${CU}_k$ matches to ${UB}_{u'}$, meaning that ${CU}_k$ is rejected at some stage in \ref{alg:step1} or \ref{alg:step2}. In conclusion, in any situation in which ${CU}_k$ is matched and ${UB}_u$ is unmatched, a blocking pair does not exist.
\item Case 2: In \textbf{$\mu_1$}, ${UB}_u$ being unmatched and ${CU}_k$ unmatched. ${UB}_u$ is unmatched means that it receives no proposal from CU, including ${CU}_k$ in both \ref{alg:step1} and \ref{alg:step2}. As both \ref{alg:step1} and \ref{alg:step2} terminate when every CU is matched or its preference list is empty. ${UB}_u$ being unmatched means that either its preference list is empty or does not contain ${UB}_u$. Then ${CU}_k$ does not have an incentive to match with ${UB}_u$. In conclusion, in any situation in which both ${CU}_k$ and ${UB}_u$ are unmatched, a blocking pair does not exist.
\item Case 3: In \textbf{$\mu_1$}, ${UB}_u$ being matched with ${CU}_k'$ and ${CU}_k$ unmatched. ${CU}_k$ is unmatched means that either it has no ${UB}_u$ in its preference list, or all its proposal have been rejected in both \ref{alg:step1}, and ${CU}_k$ can not be matched to any UBs in the reallocation stage in\ref{alg:step2}. For the former case, ${CU}_k$ does not have an incentive to match with ${UB}_u$. For the latter case, ${UB}_u$ rejects ${CU}_k$ because it prefers other proposer(s), and there are not enough spectrum resources in ${UB}_u$ to serve ${CU}_k$. Thus, ${UB}_u$ does not have incentive to match with ${CU}_k$. In conclusion, in any situation in which both ${CU}_k$ is unmatched and ${UB}_u$ is matched, a blocking pair does not exist.
\item Case 4: In \textbf{$\mu_1$}, ${UB}_u$ is matched with ${CU}_k'$ and ${CU}_k$ with ${UB}_u'$. ${UB}_u$ must be on ${CU}_k'$s preference list, and vice versa, otherwise, there is no incentive to form the (${CU}_k$, ${UB}_u$) pair. If $pri(CU_{k},UB_{u'})>pri(CU_k,UB_{u})$, then, ${CU}_k$ does not have an incentive to match with ${UB}_u$ if it is matched with ${UB}_{u'}$. If $pri(CU_{k},UB_{u})>pri(CU_k,UB_{u'})$, then, ${CU}_k$ proposes to ${UB}_u$ first and is rejected, either because ${UB}_u$ prefers ${CU}_k'$ to ${CU}_k$, or $({UB}_u, {CU}_k')$ is formed in the re-allocation stage. For the former, ${UB}_u$ does not have an incentive to match with ${CU}_k'$. For the latter, ${UB}_u$ does not have sufficient spectrum resource to serve ${CU}_k$, otherwise, the $({CU}_k, {UB}_u)$ pair has been formed in \textbf{$\mu_2$}. In conclusion, in any situation in which both ${CU}_k$ and ${UB}_u$ are matched, a blocking pair does not exist.
    \end{itemize}
Contradictions, as (${CU}_k$, ${UB}_u$) is any pair, thus, we could say that there is no blocking pair in matching \textbf{$\mu_1$}.

If there exists one blocking individual $CU_k$ or $UB_u$ in \textbf{$\mu_1$}:

for blocking individual $CU_k$:
\begin{itemize}
\item Case 1: In \textbf{$\mu_1$}, ${CU}_k$ is matched with ${UB}_u$, i.e., ${UB}_u$ is on ${CU}_k$'s preference list, then ${CU}_k$ does not have an incentive to be unmatched. In conclusion, in any situation in which both ${CU}_k$ is matched and blocking individual $CU_k$ does not exist.
    \end{itemize}
the proof that blocking individual $UB_u$ does not exist is similar to that blocking individual $CU_k$ does not exist.

In the above proof, blocking pair ($CU_k$, $UB_u$), blocking individual $CU_k$ or $UB_u$ can be  any pair or individual, thus, we could prove that there is no blocking pair or blocking individual in matching \textbf{$\mu_1$}.
\end{proof}
\begin{theorem}{Praeto optimality of \textbf{$\mu_2$}}.
In any instance of one-to-many matching, Praeto optimality is achieved by using
\ref{alg:step2} in \textbf{$\mu_2$}.
\end{theorem}

\begin{proof}
In stable matching \textbf{$\mu_1$}:
\begin{itemize}
\item Case 1: An unmatched ${CU}_k$ exists, which can be matched to ${UB}_u$ to increase the achievable EE of both ${CU}_k$ and ${UB}_u$, meaning that (${CU}_k$, ${UB}_u$) is the blocking pair of matching \textbf{$\mu_1$}, contracting \textbf{Theorem \ref{sta2}}.

\item Case 2: An existing a (${CU}_k$ exists, ${UB}_u$) pair. Obviously, ${CU}_k$ does not have an incentive to be unmatched; ${CU}_k$ has the incentive to change partner from ${UB}_u$ to ${UB}_{u'}$ to increase its achievable EE, meaning that (${UB}_u$, ${UB}_{u'}$) is a blocking pair of matching \textbf{$\mu_1$}, contracting \textbf{Theorem \ref{sta2}}.
    \end{itemize}
    It is impossible to increase the EE of a CU without decreasing that of the remaining CUs. The statement stands for UB, which can be proven similarly as above.
\end{proof}
\begin{theorem}{Complexity of \ref{alg:step2} (Convergence of \ref{alg:step2})}.
\label{complexity}
In any instance of many-to-one matching, a matching \textbf{$\mu_2$} can be obtained by using \ref{alg:step2} based on matching \textbf{$\mu_1$} within $\mathcal{O}(mU(K-m)(U-1))$ iterations, where $m$ is the number of unmatched CUs in \textbf{$\mu_1$}.
\end{theorem}
\begin{proof}
At every step in \ref{alg:step2}, each one of $m$ unmatched proposes to favourite UB, such as $UB_u$, in its current preference list. The maximum number of CUs being matched to $UB_u$ in $mu_1$ is $(K-m)$. Then, the matched CUs of $UB_u$ will be deleted from $mu_1$ and re-matched to the rest of UBs in their preference lists. The maximum number of CUs that are deleted is $(K-m)$. For each deleted CU, the maximum number of UBs in its preference list is $(U-1)$. Thus the maximum number of accepting/rejecting decisions made is $(K-m)(U-1)$ for each proposal of an unmatched CU. As there $m$ unmatched CUs, the total number of accepting/rejecting decisions made is $(K-m)(U-1)*mU$.
\end{proof}

\section{Numerical Results and Analysis}
\subsection{Simulation Setting}
\setcounter{table}{1}
\begin{table}
  \begin{center}
  \caption{Parameters for LTE-U uplink EE optimization simulation}
    \label{Parameters}
    \begin{tabular}{ | l| l |}
   \hline
   Number of CUs & 6, 9, 12, 15, 18 and 21\\
   Network Radius & 100 m\\
CU Traffic Level ($TR^C$)& 10, 15, 20, 25, 30, 35 and 40 Mbps\\
WU Traffic Level ($TR^W$)& 20 Mbps\\
Unlicensed Spectrum & 5 GHz\\
UB Bandwidth & 20 MHz\\
CU Transmission Power& 20 mw\\
T&10 $\mu$ s\\
t&1 $\mu$ s\\
  Packet Size& 12800 bits\\
 MAC header & 272 bits\\

  PHY header& 128 bits\\

 ACK & 112 bits + PHY header  \\

 Wi-Fi \& LAA Bit Rate& 50 Mbit/s\\
  ${CW}_{initial}$&8\\
 Slot Time& 9 $\mu$s\\
 SIFS&16 $\mu$s\\
 DIFS& 34 $\mu$s\\

\hline
\end{tabular}
\end{center}
\end{table}
We perform a Monte Carlo simulation in a circle with a radius of 100m, with CUs randomly and uniformly distributed being served by a SCBS. The throughput requirements of Wi-Fi users and CUs are both random values between the range of [0, $TR^W$] and [0, $TR^C$], respectively. We evaluate the performance of the proposed algorithm in the network with the number of CUs. We assume the total number of UB to be 9. We set the slot time $T$ to be 10 $\mu$s, and the sub-frame duration $t$ to be 1 $\mu$s, which is much smaller than the channel coherence time. For each scenario with a certain network density and traffic load level, simulation is run 10,000 times. CUs are randomly located in the area of interest 100 times, and in each time channel fading is performed 100 times. All other parameters can be referred to in Table. \ref{Parameters}.

\begin{figure}[htbp]
\centering\includegraphics[width=0.5\textwidth]{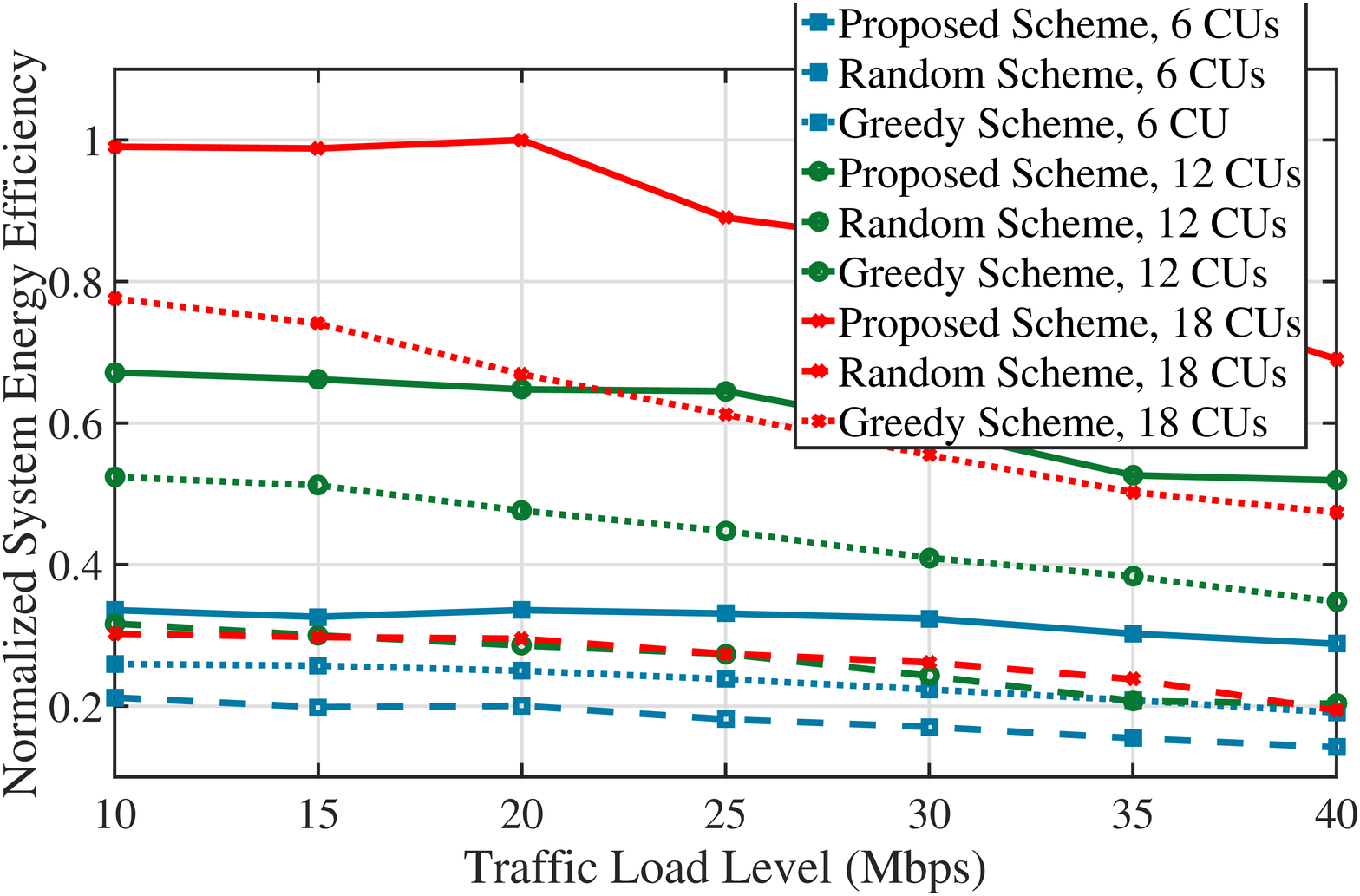}
\caption{System Energy Efficiency for Scenarios with Different Number of CUs
}\label{SEE}
\end{figure}
\subsection{Numerical Results}
\subsubsection{EE and Fairness Between CUs}
We first analyse the system EE obtained by the proposed matching-based scheme in scenarios with a different number of CUs and traffic load level in Fig. \ref{SEE}. Our proposed algorithm outperforms the greedy algorithm and random allocation under both low-density (6 CUs) and high-density networks (18 CUs) with a light traffic load from 10 Mbps per CU and heavy traffic load at 40 Mbps per CU. The system EE improves 30\% and 50\% obtained by our proposed method as compared with that obtained by the greedy algorithm, under the light and the heavy traffic load scenarios respectively. For the same number of CUs, with the increasing of traffic load per CU, the system EE decreases because more CUs remain unserved in the heavy traffic load scenario, as shown in Fig. \ref{SCU}. This is because more resources are occupied to serve a CU with a high traffic demand, leading to a drop in the number of CUs that can served in the network, i.e., more CUs fail to achieve their throughput requirement.

\begin{figure}[htbp]
\centering\includegraphics[width=0.5\textwidth]{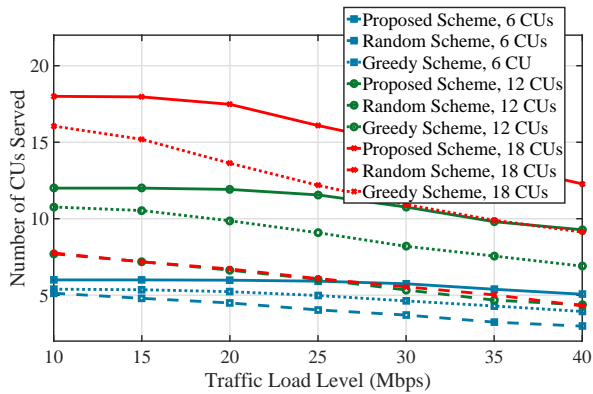}
\caption{The Number of CUs Served}
\label{SCU}
\end{figure}

\begin{figure}[htbp]
\centering\includegraphics[width=0.5\textwidth]{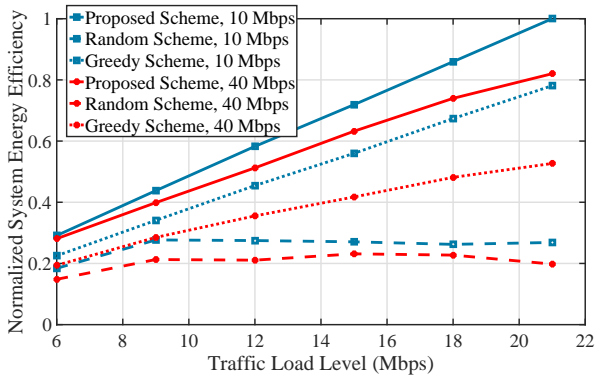}
\caption{System Energy Efficiency in Different Traffic Load Level}
\label{SEET}
\end{figure}
On the contrary, with the same traffic load level, more CUs tend to be served in the dense scenarios, leading to an increase of system EE as shown in Fig. \ref{SEET}. In dense scenario, more CUs have the chance to meet their throughput requirement, due to many factors, such as the distance between CU and SCBS and channel condition between CU and SCBS. Although the number of CUs served increases with the number of CUs in the network, except for the low traffic demand scenario, the percentage of CUs that have their throughput requirement fulfilled drops, as shown in Fig. \ref{CUP}. In a low traffic demand scenario, where the spectrum resource is sufficient to serve every CU with their required throughput demand, almost 100\% of CUs' being served rate is achieved by the proposed algorithm, compared with less than 90\% achieved by the greedy algorithm and the even lower served rate when using a random algorithm. In medium and high traffic demand scenario, the percentage of CUs served decreases with the increase of CUs in the network by using any one of the three algorithms. However, the proposed algorithm still outperforms the greedy algorithm and random algorithm by around 35\% and 50\%~120\%, respectively. Thus, we could say that the proposed algorithm works more effectively in CUs' fairness compared with the greedy algorithm or the random allocation scheme.

\begin{figure}[htbp]
\centering\includegraphics[width=0.5\textwidth]{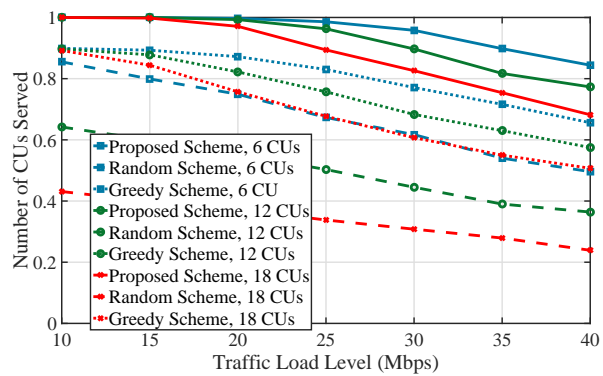}
\caption{The Percentage of CUs Served Comparison}
\label{CUP}
\end{figure}

\subsubsection{Throughput Analysis}
Throughput is another performance matrix for both the system and an individual CU. As shown in Fig. \ref{STP}, in the 6 CUs scenarios with  low traffic demand, three algorithms achieve similar results. This is because the unlicensed spectrum resource is sufficient to serve every CU with their relatively low traffic demands. In low traffic demand, system throughput increases with the number of CUs almost linearly as shown by using the proposed algorithm and the greedy algorithm, because the spectrum resource is still sufficient. The proposed algorithm outperforms the greedy algorithm. However, there is another aspect in heavy traffic load. In the network with 6 CUs, the proposed algorithm achieves 66\% more than the greedy algorithm, and more than 100\% more than the random scheme. With the increase of the number of CUs in the network, the overall throughput achieved by using the proposed algorithm tends to saturate in heavy traffic load scenarios. This is because the capacity is limited by the available unlicensed spectrum resources.

\begin{figure}[htbp]
\centering\includegraphics[width=0.5\textwidth]{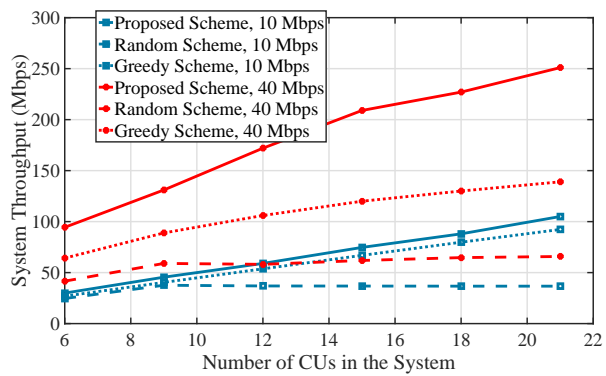}
\caption{System Throughput In Different Traffic load Level}
\label{STP}
\end{figure}

\subsubsection{Computational Complexity}

The theoretical upper bound of the computation complexity of \ref{alg:step1} and \ref{alg:step2} have been given in \textbf{Theorem \ref{comSTEP1}}, and \textbf{Theorem \ref{complexity}}. Here we show the actual computation complexity of the proposed algorithm in typical traffic load scenarios in Fig. \ref{fig:complexity}.

There are positive correlations between the complexity and network density at the same traffic load level. Specifically, at the lowest traffic load (10 Mbps), complexity is slightly more than the number of CUs in the network. This means that almost all the CUs' first proposal are accepted, due to the low traffic demand of each CU. In a low traffic case, most CUs are matched by using $\ref{alg:step1}$; $\ref{alg:step2}$ is seldom performed. The complexity increases with the traffic load level from 10 to 30 Mbps. This is because with the increase of traffic load level, increasing CUs are unmatched in \textbf{$\mu_1$} by using $\ref{alg:step1}$; the number of iterations that $\ref{alg:step2}$ performs is increasing. The complexity of an iteration in $\ref{alg:step2}$ ($\mathcal{O}((K-m)(U-1))$) is much larger than that in $\ref{alg:step1}$ ($\mathcal{O}(U)$), leading to an increase of complexity. At an even higher traffic load level, the complexity begins to drop. At this stage, the number of UBs in a CU's preference lists is much smaller than that in a medium traffic load level. The complexity of obtaining matching \textbf{$\mu_1$} is much smaller. Although the number of unmatched CUs rises in the scenario with the same network density, elements in their preference lists are much smaller, the complexity in an iteration drops significantly, leading to the decrease of computational complexity at a high traffic load level.

\begin{figure}[htbp]
\centering\includegraphics[width=0.5\textwidth]{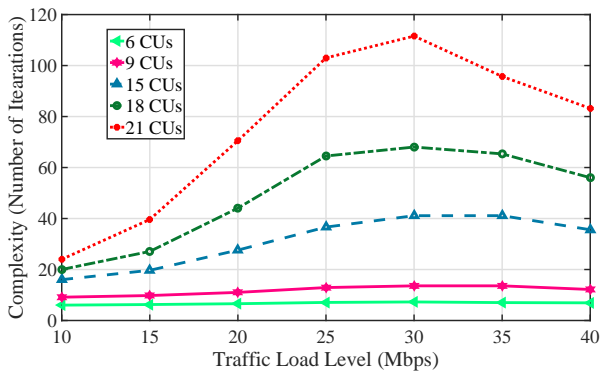}
\caption{Computational Complexity in Different Scenario}
\label{fig:complexity}
\end{figure}

\section{Conclusion}

In this work, we have studied the uplink resource allocation problem in a LTE-U and Wi-Fi coexistence scenario to maximize each CU's EE. We formulated the problem as a multi-objective optimization, and transformed it into a single-objective optimization by using the weighted-sum method. We proposed a semi-distributed 2-step matching with partial information based algorithm to solve the problem. Compared with the greedy algorithm based resource allocation scheme, our proposed scheme achieves improvements of up to $50\%$ in terms of EE and up to $66\%$ in terms of throughput. Furthermore, we have analysed the computational complexity of the proposed algorithm theoretically and by simulations, thereby showing the complexity is reasonable for real-world deployment.

In the future, work will be extended into the heterogeneous LTE-U networks, where hyper-dense deployment of LTE-U cells may exist. We will also consider a comprehensive optimized resource allocation scheme for LTE-U taking into account that CU can choose between licensed and unlicensed bands. In such scenarios, a multi-side matching model should be considered, which poses new challenges in achieving the solutions.
\section*{Acknowledgment}

This paper acknowledges the support of the MOST of China for the "Small Cell and Heterogeneous Network Planning and Deployment" project under grant No. 2015DFE12820,  and H2020 DECADE project.
\bibliographystyle{IEEEtran}
\bibliography{Myreference}

\begin{thebibliography}{10}
\providecommand{\url}[1]{#1}
\csname url@samestyle\endcsname
\providecommand{\newblock}{\relax}
\providecommand{\bibinfo}[2]{#2}
\providecommand{\BIBentrySTDinterwordspacing}{\spaceskip=0pt\relax}
\providecommand{\BIBentryALTinterwordstretchfactor}{4}
\providecommand{\BIBentryALTinterwordspacing}{\spaceskip=\fontdimen2\font plus
\BIBentryALTinterwordstretchfactor\fontdimen3\font minus
  \fontdimen4\font\relax}
\providecommand{\BIBforeignlanguage}[2]{{%
\expandafter\ifx\csname l@#1\endcsname\relax
\typeout{** WARNING: IEEEtran.bst: No hyphenation pattern has been}%
\typeout{** loaded for the language `#1'. Using the pattern for}%
\typeout{** the default language instead.}%
\else
\language=\csname l@#1\endcsname
\fi
#2}}
\providecommand{\BIBdecl}{\relax}
\BIBdecl

\bibitem{bleicher2013surge}
A.~Bleicher, ``A surge in small cells [2013 tech to watch],'' \emph{IEEE
  Spectrum}, vol.~50, no.~1, pp. 38--39, 2013.

\bibitem{chen2017coexistence}
B.~Chen, J.~Chen, Y.~Gao, and J.~Zhang, ``{Coexistence of LTE-LAA and Wi-Fi on
  5 GHz with corresponding deployment scenarios: A survey},'' \emph{IEEE
  Communications Surveys \& Tutorials}, vol.~19, no.~1, pp. 7--32, 2017.

\bibitem{wang2016evolution}
M.~Wang, J.~Zhang, B.~Ren, W.~Yang, J.~Zou, M.~Hua, and X.~You, ``{The
  evolution of LTE physical layer control channels},'' \emph{IEEE
  Communications Surveys \& Tutorials}, vol.~18, no.~2, pp. 1336--1354, 2016.

\bibitem{lee2014recent}
Y.~L. Lee, T.~C. Chuah, J.~Loo, and A.~Vinel, ``Recent advances in radio
  resource management for heterogeneous {LTE/LTE-A} networks,'' \emph{IEEE
  Communications Surveys \& Tutorials}, vol.~16, no.~4, pp. 2142--2180, 2014.

\bibitem{ku2015resource}
G.~Ku and J.~M. Walsh, ``{Resource allocation and link adaptation in LTE and
  LTE advanced: A tutorial},'' \emph{IEEE Communications Surveys \& Tutorials},
  vol.~17, no.~3, pp. 1605--1633, 2015.

\bibitem{yued2dlteu}
Y.~Wu, W.~Guo, H.~Yuan, L.~Li, S.~Wang, X.~Chu, and J.~Zhang,
  ``{Device-to-device meets LTE-unlicensed},'' \emph{IEEE Communications
  Magazine}, vol.~54, no.~5, pp. 154--159, May 2016.

\bibitem{babaei2015impact}
A.~Babaei, J.~Andreoli-Fang, Y.~Pang, and B.~Hamzeh, ``{On the impact of LTE-U
  on Wi-Fi performance},'' \emph{International Journal of Wireless Information
  Networks}, vol.~22, no.~4, pp. 336--344, 2015.

\bibitem{rupasinghe2014licensed}
N.~Rupasinghe and {\.I}.~G{\"u}ven{\c{c}}, ``{Licensed-assisted access for
  WiFi-LTE coexistence in the unlicensed spectrum},'' in \emph{Globecom
  Workshops (GC Wkshps), 2014}.\hskip 1em plus 0.5em minus 0.4em\relax IEEE,
  2014, pp. 894--899.

\bibitem{alcatel2016verizon}
E.~Alcatel-Lucent and S.~Qualcomm, ``{Verizon,¡°LTE-U Technical Report
  Coexistence Study for LTE-U SDL V1. 0,¡± LTE-U Forum},'' Tech. Rep., Feb.
  2015. Accessed on 10/06, Tech. Rep., 2016.

\bibitem{zhang2015coexistence}
H.~Zhang, X.~Chu, W.~Guo, and S.~Wang, ``Coexistence of wi-fi and heterogeneous
  small cell networks sharing unlicensed spectrum,'' \emph{IEEE Communications
  Magazine}, vol.~53, no.~3, pp. 158--164, 2015.

\bibitem{chen2015downlink}
C.~Chen, R.~Ratasuk, and A.~Ghosh, ``{Downlink performance analysis of LTE and
  WiFi coexistence in unlicensed bands with a simple listen-before-talk
  scheme},'' in \emph{Vehicular Technology Conference (VTC Spring), 2015 IEEE
  81st}.\hskip 1em plus 0.5em minus 0.4em\relax IEEE, 2015, pp. 1--5.

\bibitem{ETSI1.7.1}
\emph{{ETSI EN 301 893 V1.7.1}}, ETSI Std., June 2012.

\bibitem{al20155g}
A.~Al-Dulaimi, S.~Al-Rubaye, Q.~Ni, and E.~Sousa, ``{5G communications race:
  Pursuit of more capacity triggers LTE in unlicensed band},'' \emph{IEEE
  vehicular technology magazine}, vol.~10, no.~1, pp. 43--51, 2015.

\bibitem{chiang2005geometric}
M.~Chiang \emph{et~al.}, ``Geometric programming for communication systems,''
  \emph{Foundations and Trends{\textregistered} in Communications and
  Information Theory}, vol.~2, no. 1--2, pp. 1--154, 2005.

\bibitem{chen2016cellular}
Q.~Chen, G.~Yu, H.~Shan, A.~Maaref, G.~Y. Li, and A.~Huang, ``{Cellular meets
  WiFi: Traffic offloading or resource sharing?}'' \emph{IEEE Transactions on
  Wireless Communications}, vol.~15, no.~5, pp. 3354--3367, 2016.

\bibitem{ko2016fair}
H.~Ko, J.~Lee, and S.~Pack, ``{A fair listen-before-talk algorithm for
  coexistence of LTE-U and WLAN},'' \emph{IEEE Transactions on Vehicular
  Technology}, vol.~65, no.~12, pp. 10\,116--10\,120, 2016.

\bibitem{gu2015exploiting}
Y.~Gu, Y.~Zhang, L.~X. Cai, M.~Pan, L.~Song, and Z.~Han, ``{Exploiting
  student-project allocation matching for spectrum sharing in
  LTE-unlicensed},'' in \emph{Global Communications Conference (GLOBECOM), 2015
  IEEE}.\hskip 1em plus 0.5em minus 0.4em\relax IEEE, 2015, pp. 1--6.

\bibitem{gao}
X.~C. Y.~Gao, B~Chen and J.~Zhang, ``{Resource Allocation in LTE-LAA and WiFi
  Coexistence: A Joint Contention Window Optimization Scheme},'' in
  \emph{Globecom (GC), 2017}.\hskip 1em plus 0.5em minus 0.4em\relax IEEE,
  2017.

\bibitem{hamidouche2016multi}
K.~Hamidouche, W.~Saad, and M.~Debbah, ``{Multi-Games for LTE and WiFi
  Coexistence over Unlicensed Channels},'' in \emph{International Conference on
  Network Games, Control, and Optimization}.\hskip 1em plus 0.5em minus
  0.4em\relax Springer, 2016, pp. 123--133.

\bibitem{cano2015coexistence}
C.~Cano and D.~J. Leith, ``{Coexistence of WiFi and LTE in unlicensed bands: A
  proportional fair allocation scheme},'' in \emph{Communication Workshop
  (ICCW), 2015 IEEE International Conference on}.\hskip 1em plus 0.5em minus
  0.4em\relax IEEE, 2015, pp. 2288--2293.

\bibitem{sagari2015coordinated}
S.~Sagari, S.~Baysting, D.~Saha, I.~Seskar, W.~Trappe, and D.~Raychaudhuri,
  ``{Coordinated dynamic spectrum management of LTE-U and Wi-Fi networks},'' in
  \emph{Dynamic Spectrum Access Networks (DySPAN), 2015 IEEE International
  Symposium on}.\hskip 1em plus 0.5em minus 0.4em\relax IEEE, 2015, pp.
  209--220.

\bibitem{ni2012nash}
Q.~Ni and C.~C. Zarakovitis, ``Nash bargaining game theoretic scheduling for
  joint channel and power allocation in cognitive radio systems,'' \emph{IEEE
  Journal on selected areas in Communications}, vol.~30, no.~1, pp. 70--81,
  2012.

\bibitem{chen2016optimizing}
Q.~Chen, G.~Yu, and Z.~Ding, ``{Optimizing unlicensed spectrum sharing for
  LTE-U and WiFi network coexistence},'' \emph{IEEE Journal on Selected Areas
  in Communications}, vol.~34, no.~10, pp. 2562--2574, 2016.

\bibitem{aliresource}
M.~Ali, S.~Qaisar, M.~Naeem, J.~J. Rodrigues, and F.~Qamar, ``Resource
  allocation for licensed and unlicensed spectrum in 5g heterogeneous
  networks,'' \emph{Transactions on Emerging Telecommunications Technologies},
  p. e3299.

\bibitem{nielsen2014lte}
S.~Nielsen and A.~Toskala, ``{LTE in unlicensed spectrum: European regulation
  and co-existence considerations},'' in \emph{3GPP workshop on LTE in
  unlicensed spectrum}, 2014.

\bibitem{garnaev2016fair}
A.~Garnaev, S.~Sagari, and W.~Trappe, ``{Fair channel sharing by Wi-Fi and
  LTE-U networks with equal priority},'' in \emph{International Conference on
  Cognitive Radio Oriented Wireless Networks}.\hskip 1em plus 0.5em minus
  0.4em\relax Springer, 2016, pp. 91--103.

\bibitem{yang2017interference}
C.~Yang, J.~Li, Q.~Ni, A.~Anpalagan, and M.~Guizani, ``{Interference-aware
  energy efficiency maximization in 5G ultra-dense networks},'' \emph{IEEE
  Transactions on Communications}, vol.~65, no.~2, pp. 728--739, 2017.

\bibitem{feng2013survey}
D.~Feng, C.~Jiang, G.~Lim, L.~J. Cimini, G.~Feng, and G.~Y. Li, ``A survey of
  energy-efficient wireless communications,'' \emph{IEEE Communications Surveys
  \& Tutorials}, vol.~15, no.~1, pp. 167--178, 2013.

\bibitem{li2011energy}
G.~Y. Li, Z.~Xu, C.~Xiong, C.~Yang, S.~Zhang, Y.~Chen, and S.~Xu,
  ``Energy-efficient wireless communications: tutorial, survey, and open
  issues,'' \emph{IEEE Wireless Communications}, vol.~18, no.~6, 2011.

\bibitem{gu2015matching}
Y.~Gu, W.~Saad, M.~Bennis, M.~Debbah, and Z.~Han, ``{Matching theory for future
  wireless networks: fundamentals and applications},'' \emph{IEEE
  Communications Magazine}, vol.~53, no.~5, pp. 52--59, 2015.

\bibitem{di2016joint}
B.~Di, S.~Bayat, L.~Song, Y.~Li, and Z.~Han, ``Joint user pairing, subchannel,
  and power allocation in full-duplex multi-user {OFDMA} networks,'' \emph{IEEE
  Transactions on Wireless Communications}, vol.~15, no.~12, pp. 8260--8272,
  2016.

\bibitem{sekander2017decoupled}
S.~Sekander, H.~Tabassum, and E.~Hossain, ``{Decoupled uplink-downlink user
  association in multi-tier full-duplex cellular networks: A two-sided matching
  game},'' \emph{IEEE Transactions on Mobile Computing}, vol.~16, no.~10, pp.
  2778--2791, 2017.

\bibitem{zhao2017many}
J.~Zhao, Y.~Liu, K.~K. Chai, Y.~Chen, and M.~Elkashlan, ``{Many-to-Many
  Matching With Externalities for Device-to-Device Communications},''
  \emph{IEEE Wireless Communications Letters}, vol.~6, no.~1, pp. 138--141,
  2017.

\bibitem{gu2017dynamic}
Y.~Gu, C.~Jiang, L.~X. Cai, M.~Pan, L.~Song, and Z.~Han, ``{Dynamic Path To
  Stability in LTE-Unlicensed with User Mobility: A Matching Framework},''
  \emph{IEEE Transactions on Wireless Communications}, 2017.

\bibitem{bianchi2000performance}
G.~Bianchi, ``Performance analysis of the ieee 802.11 distributed coordination
  function,'' \emph{IEEE Journal on selected areas in communications}, vol.~18,
  no.~3, pp. 535--547, 2000.

\bibitem{ruzika2005approximation}
S.~Ruzika and M.~M. Wiecek, ``Approximation methods in multiobjective
  programming,'' \emph{Journal of optimization theory and applications}, vol.
  126, no.~3, pp. 473--501, 2005.

\bibitem{manlove2013algorithmics}
D.~F. Manlove, \emph{Algorithmics of matching under preferences}.\hskip 1em
  plus 0.5em minus 0.4em\relax World Scientific, 2013, vol.~2.

\bibitem{iwama2008survey}
K.~Iwama and S.~Miyazaki, ``A survey of the stable marriage problem and its
  variants,'' in \emph{Informatics Education and Research for
  Knowledge-Circulating Society, 2008. ICKS 2008. International Conference
  on}.\hskip 1em plus 0.5em minus 0.4em\relax IEEE, 2008, pp. 131--136.

\bibitem{shrivastava2016stable}
A.~Shrivastava and C.~P. Rangan, ``{Stable marriage problem with ties and
  incomplete bounded length preference list under social stability},''
  \emph{arXiv preprint arXiv:1601.03523}, 2016.

\end{thebibliography}

\end{document}